\documentclass{article}%
\usepackage{amsfonts}
\usepackage{amsmath}
\usepackage{fancybox}
\usepackage{fancyhdr}
\usepackage{color}
\usepackage{graphicx}%
\usepackage{amssymb}
\providecommand{\U}[1]{\protect\rule{.1in}{.1in}}
\newtheorem{theorem}{Theorem}

\newtheorem{corollary}[theorem]{Corollary}

\newenvironment{proof}[1][Proof]{\noindent\textbf{#1.} }{\ \rule{0.5em}{0.5em}}
\begin{document}

\title{Closed form solution for the surface area, the capacitance and the
demagnetizing factors of the ellipsoid}
\author{G. V. Kraniotis\thanks{Email: gkraniot@cc.uoi.gr}, G. K.
Leontaris\thanks{Emai: leonta@uoi.gr}\\University of Ioannina, Department of Physics, \\Section of Theoretical Physics,\\Ioannina GR- 451 10, Greece}
\maketitle

\begin{abstract}
We derive the closed form solutions for the surface area, the capacitance and
the demagnetizing factors of the ellipsoid immersed in the Euclidean space $%
\mathbb{R}
^{3}$. The exact solutions \ for the above geometrical and physical properties
of the ellipsoid are expressed \ elegantly in terms of the generalized
hypergeometric functions of Appell of two variables. Various limiting cases of
the theorems of the exact solution for the surface area, the demagnetizing
factors and the capacitance of the ellipsoid are derived, which agree with
known solutions for the prolate and oblate spheroids and the sphere. Possible
applications of the results achieved, in various fields of science, such as in
physics, biology and space science are briefly discussed.

\end{abstract}

\section{Introduction}

An interesting and important problem of geometry and mathematical analysis is
the exact answer to the question: which is the surface area of the ellipsoid
immersed in the Euclidean space $%
\mathbb{R}
^{3}$? Despite the simplicity of the question and the fact that the roots of
the problem can be traced back to the 19th Century there has been only a
partial progress towards its solution. This is because the closed form
solution had evaded the efforts of previous researchers and scholars. The
first serious investigation had been performed by Legendre who obtained an
equation for the surface area of the ellipsoid in terms of formal integrals
\cite{AMLEGENDRE}. At this point we note, that a nice and critical review of
the mathematical literature summarizing the attempts of various mathematicians
in solving the problem, from the period of Legendre till 2005, has been
written in \cite{TEE} (see for instance \cite{Keller} $\ $cited in
\cite{TEE}). There is also a practical interest for an exact solution for the
ellipsoidal surface area in various fields of science, we just mention a few
such fields: 1) in biology the human cornea as well as the chicken
erythrocytes are realistically described by an ellipsoid and the area is
important in the latter case for the determination of the permeabilities of
the cells \cite{KWOK},\cite{brian} 2) in cosmology and the physics of rotating
black holes \cite{Krasinski}and 3) in the geometry of hard ellipsoidal
molecules and their virial coefficients. In particular in the latter case, the
surface area appears in the expression for the pressure of the ellipsoidal
molecules \cite{Singh}. We also mention the relevance of the surface area of
ellipsoid for the investigation and measurement of capillary forces between
sediment particles and an air-water interface \cite{Chatterjee}. For an
application to medicine we refer the reader to \cite{Xu}.

On the other hand there are two further important aspects related to the
geometry of the ellipsoid awaiting for a full analytic solution with many
important applications. Namely:\textit{\ first} the calculation in closed
analytic form of the \textit{capacitance }of a conducting ellipsoid and
\textit{second} the exact analytic calculation of the \textit{demagnetizing
factors} of a magnetized ellipsoid.

In the former case, the geometry of the ellipsoid is complex enough to serve
as a promising avenue for modeling arbitrarily shaped conducting bodies
\cite{shumpert}. Capacitance modulation has been suggested recently as a
method of detecting microorganisms such as the E. coli \ present in the water
\cite{Dwivedi}. Despite its importance in theory and applications, no exact
analytic solution for the capacitance of the ellipsoid \ had been derived by
previous authors. There was only a formula in terms of formal integrals
derived in \cite{shumpert}.

In the later case, the magnetic susceptibility $\kappa$ of the body determined
in the ambient magnetic field $\vec{B}$ is influenced by the shape and
dimensions of the body. Thus the measured (apparent) magnetic susceptibility
$\kappa_{A}$ should be corrected for this shape effect to obtain the
shape-independent true susceptibility $\kappa_{T}.$ The relation between the
true and apparent volume susceptibility involves the so called demagnetizing
factors. The first attempts of calculating the demagnetizing factors \ of the
ellipsoid were made in \cite{Kellogg},\cite{OSBORN}. However, the authors of
these works only derived expressions in terms of formal integrals. \ In this
paper, we derive for the first time the closed form solution for the three
demagnetizing factors for the ellipsoid, in terms of the first hypergeometric
function of Appell of two variables. A fundamental application of our work
will be in the determination of asteroidal magnetic susceptibility and its
comparison to those of meteorites in order to establish a meteorite-asteroid
match \cite{Kohout}. Another interesting application of our solution for the
demagnetizing factors of the ellipsoid would be in the field of microrobots.
An external magnetic field can induce torque on a ferromagnetic body. Thus the
use of external magnetic fields has strong advantages in microrobotics and
biomedicine such as wireless controllability and safe use in clinical
applications \cite{tottori}.

Thus, there is a certain demand from pure and applied mathematics for the
closed form solutions of the above geometric problems. It is the purpose of
our paper to produce such novel and useful exact analytic solutions for all
three described problems above. We report our findings in what follows.

\section{Closed form solution for the surface area of the ellipsoid.}

We consider an ellipsoid centred at the coordinate origin, with rectangular
Cartesian coordinate axes along the semi-axes $a,b,c:$%
\begin{equation}
\frac{x^{2}}{a^{2}}+\frac{y^{2}}{b^{2}}+\frac{z^{2}}{c^{2}}%
=1.\label{ellipsoidgeniko}%
\end{equation}
We begin our exact analytic calculation for the infinitesimal surface area
\textrm{d}$S,$ using the formula for the surface $\mathrm{Monge}$ segment :%
\begin{equation}
\vec{x}(x,y)=(x,y,z(x,y)),
\end{equation}%
\begin{equation}
\mathrm{d}S=\left\vert \vec{x}_{x}\times\vec{x}_{y}\right\vert \mathrm{d}%
x\mathrm{d}y=\sqrt[2]{1+z_{x}^{2}+z_{y}^{2}}\mathrm{d}x\mathrm{d}%
y,\label{gensurfarea}%
\end{equation}
where $z_{x}:=\frac{\partial z(x.y)}{\partial x},$ $z_{y}:=\frac{\partial
z(x,y)}{\partial y}$ and$,$
\begin{align}
1+z_{x}^{2}+z_{y}^{2}  & =1+x^{2}\left(  \frac{c^{2}}{a^{2}}\right)  ^{2}%
\frac{1}{c^{2}\left(  1-\frac{x^{2}}{a^{2}}-\frac{y^{2}}{b^{2}}\right)
}+y^{2}\left(  \frac{c^{2}}{b^{2}}\right)  ^{2}\frac{1}{c^{2}\left(
1-\frac{x^{2}}{a^{2}}-\frac{y^{2}}{b^{2}}\right)  }\nonumber\\
& =1+x^{2}\left(  \frac{c^{2}}{a^{4}}\right)  \frac{1}{\left(  1-\frac{x^{2}%
}{a^{2}}-\frac{y^{2}}{b^{2}}\right)  }+y^{2}\left(  \frac{c^{2}}{b^{4}%
}\right)  \frac{1}{\left(  1-\frac{x^{2}}{a^{2}}-\frac{y^{2}}{b^{2}}\right)
}\nonumber\\
& =\frac{1-\frac{x^{2}}{a^{2}}-\frac{y^{2}}{b^{2}}+\frac{x^{2}c^{2}}{a^{4}%
}+\frac{y^{2}c^{2}}{b^{4}}}{1-\frac{x^{2}}{a^{2}}-\frac{y^{2}}{b^{2}}}%
=\frac{1-\left(  1-\frac{c^{2}}{a^{2}}\right)  \frac{x^{2}}{a^{2}}-\left(
1-\frac{c^{2}}{b^{2}}\right)  \frac{y^{2}}{b^{2}}}{1-\frac{x^{2}}{a^{2}}%
-\frac{y^{2}}{b^{2}}}\implies\nonumber\\
&
\end{align}
Substituting to $(\ref{gensurfarea})$ we get%
\begin{align}
\mathrm{d}S  & =\sqrt[2]{\frac{1-\left(  1-\frac{c^{2}}{a^{2}}\right)
\frac{x^{2}}{a^{2}}-\left(  1-\frac{c^{2}}{b^{2}}\right)  \frac{y^{2}}{b^{2}}%
}{1-\frac{x^{2}}{a^{2}}-\frac{y^{2}}{b^{2}}}}\mathrm{d}x\mathrm{d}y\nonumber\\
& =\sqrt[2]{\frac{1-\delta\frac{x^{2}}{a^{2}}-\varepsilon\frac{y^{2}}{b^{2}}%
}{1-\frac{x^{2}}{a^{2}}-\frac{y^{2}}{b^{2}}}}\mathrm{d}x\mathrm{d}y,
\end{align}
where we define:%
\begin{equation}
\delta:=1-c^{2}/a^{2},\text{ }\varepsilon:=1-c^{2}/b^{2}.
\end{equation}
Consequently the octant surface area is given by$:$%
\begin{align}
\mathcal{A}_{oct.}^{ellipsoid}  & =\int_{0}^{a}\left\{  \int_{0}%
^{b\sqrt[2]{1-x^{2}/a^{2}}}\sqrt[2]{\frac{1-\delta\frac{x^{2}}{a^{2}%
}-\varepsilon\frac{y^{2}}{b^{2}}}{1-\frac{x^{2}}{a^{2}}-\frac{y^{2}}{b^{2}}}%
}\mathrm{d}y\right\}  \mathrm{d}x\nonumber\\
& =\int_{0}^{a}\left\{  \int_{0}^{b\sqrt[2]{1-x^{2}/a^{2}}}\sqrt{\frac{\left(
1-\delta\frac{x^{2}}{a^{2}}\right)  \left[  1-\frac{\varepsilon}{b^{2}\left(
1-\delta\frac{x^{2}}{a^{2}}\right)  }y^{2}\right]  }{\left(  1-\frac{x^{2}%
}{a^{2}}\right)  \left[  1-\frac{1}{b^{2}}\frac{1}{1-\frac{x^{2}}{a^{2}}}%
y^{2}\right]  }}\mathrm{d}y\right\}  \mathrm{d}x\nonumber\\
& =\int_{0}^{a}\left\{  \int_{0}^{b\sqrt[2]{1-x^{2}/a^{2}}}\Omega\sqrt
{\frac{(1-\mu^{2}y^{2})}{(1-\lambda^{2}y^{2}}}\mathrm{d}y\right\}
\mathrm{d}x,
\end{align}
with%
\begin{equation}
\Omega:=\sqrt[2]{\frac{\left(  1-\delta\frac{x^{2}}{a^{2}}\right)  }{\left(
1-\frac{x^{2}}{a^{2}}\right)  }},\text{ }\mu^{2}:=\frac{\varepsilon}%
{b^{2}(1-\delta x^{2}/a^{2})},\lambda^{2}:=\frac{1}{b^{2}}\frac{1}%
{1-\frac{x^{2}}{a^{2}}}.
\end{equation}
We define a new variable:%
\begin{equation}
y_{1}:=\frac{y}{b\sqrt[2]{1-x^{2}/a^{2}}}\implies\mathrm{d}y_{1}%
=\frac{\mathrm{d}y}{\eta},\text{ with }\eta:=b\sqrt[2]{1-x^{2}/a^{2}}.
\end{equation}
Thus:%
\begin{align}
\mathcal{A}_{oct.}^{ellipsoid}  & =\int_{0}^{a}\left\{  \int_{0}^{1}%
\Omega(x)\sqrt{\frac{1-\mu^{\prime2}y_{1}^{2}}{1-\lambda^{\prime2}y_{1}^{2}}%
}\mathrm{d}y_{1}\eta\right\}  \mathrm{d}x\nonumber\\
& =\int_{0}^{a}\Omega(x)\eta(x)\left\{  \int_{0}^{1}\sqrt{\frac{1-\mu
^{\prime2}\psi}{1-\lambda^{\prime2}\psi}}\frac{\mathrm{d}\psi}{2\sqrt[2]{\psi
}}\right\}  \mathrm{d}x\implies
\end{align}%
\begin{equation}
\shadowbox{$\displaystyle\mathcal{A}_{o\kappa\tau.}^{ellipsoid}=\int
_0 ^a\Omega(x)\eta(x)\frac{1}{2}\frac{\Gamma(1/2)\Gamma(1)}{\Gamma
(3/2)}F_1\left(\frac{1}{2},-\frac{1}{2},\frac{1}{2},\frac{3}{2},\mu^{\prime
2},\lambda^{\prime2}\right){\rm d}x,
\label{KRANIOTISELLIPS} $}
\end{equation}
with
\begin{align}
\mu^{\prime2}  & :=\eta^{2}\mu^{2}=\frac{(1-x^{2}/a^{2})}{(1-\delta
x^{2}/a^{2})}(1-c^{2}/b^{2}),\\
\lambda^{\prime2}  & :=\eta^{2}\frac{1}{b^{2}(1-x^{2}/a^{2})}=\frac
{b^{2}(1-x^{2}/a^{2})}{b^{2}(1-x^{2}/a^{2})}=1
\end{align}
and $F_{1}(\alpha,\beta,\beta^{\prime},\gamma,x,y)$ denotes the first
generalized hypergeometric function of \textrm{Appell \cite{appell} }with two
variables $x,y$ and parameters $\alpha,\beta,\beta^{\prime},\gamma:$%
\begin{equation}
\shadowbox{$\displaystyle F_1(\alpha,\beta,\beta^{\prime},\gamma
,x,y)=\sum_{m=0}^{\infty}\sum_{n=0}^{\infty}
\frac{(\alpha,m+n)(\beta,m)(\beta^{\prime},n)}{(\gamma,m+n)(1,m)(1,n)}%
x^m y^n.$}
\label{AppellGVKellipsoid}
\end{equation}
The double series converges absolutely for $|x|<1,|y|<1.$

Thus we obtain $:$%
\begin{align}
\mathcal{A}_{oct.}^{ellipsoid}  & =\int_{0}^{a}b\sqrt[2]{1-\frac{\delta x^{2}%
}{a^{2}}}\frac{1}{2}\frac{\Gamma(1/2)\Gamma(1)}{\Gamma(3/2)}F_{1}\left(
\frac{1}{2},-\frac{1}{2},\frac{1}{2},\frac{3}{2},\mu^{\prime2},1\right)
\mathrm{d}x\nonumber\\
& =\int_{0}^{a}b\sqrt[2]{1-\frac{\delta x^{2}}{a^{2}}}\frac{1}{2}\frac
{\Gamma(1/2)\Gamma(1)}{\Gamma(3/2)}\frac{\Gamma(3/2)\Gamma(1/2)}%
{\Gamma(1)\Gamma(1)}F\left(  \frac{1}{2},-\frac{1}{2},1,\mu^{\prime2}\right)
\mathrm{d}x\nonumber\\
& =\int_{0}^{a}b\sqrt[2]{1-\frac{\delta x^{2}}{a^{2}}}\frac{\pi}{2}F\left(
\frac{1}{2},-\frac{1}{2},1,\frac{(1-x^{2}/a^{2})}{(1-\delta x^{2}/a^{2}%
)}(1-c^{2}/b^{2})\right)  \mathrm{d}x.\nonumber\\
&
\end{align}
In the transition from the first to the second line of the previous equation
we made use of the property of Appell's hypergeometric function according to
which if one of its two variables is set to the value $1$ (one), \ then the
function $F_{1}$ reduces to the ordinary hypergeometric function of Gau\ss :
\begin{equation}
F_{1}\left(  \frac{1}{2},-\frac{1}{2},\frac{1}{2},\frac{3}{2},\mu^{\prime
2},1\right)  =\frac{\Gamma(3/2)\Gamma(3/2-1/2-1/2)}{\Gamma(3/2-1/2)\Gamma
(3/2-1/2)}F\left(  \frac{1}{2},-\frac{1}{2},\frac{3}{2}-\frac{1}{2}%
,\mu^{\prime2}\right)  .
\end{equation}
It is also valid$:$%
\begin{equation}
\mathcal{A}_{oct.}^{ellipsoid}=ab\int_{0}^{1}\sqrt[2]{1-\delta x_{1}^{2}}%
\frac{\pi}{2}F\left(  \frac{1}{2},-\frac{1}{2},1,\frac{(1-x_{1}^{2}%
)}{(1-\delta x_{1}^{2})}(1-c^{2}/b^{2})\right)  \mathrm{d}x_{1}.
\end{equation}
We now apply the transformation$:$%
\begin{equation}
\frac{1-x_{1}^{2}}{1-\delta x_{1}^{2}}=1-\eta^{2}%
\end{equation}
which yields$:$%
\begin{align}
\mathcal{A}_{oct.}^{ellipsoid}  & =\frac{ab}{\sqrt[2]{1-\delta}}\int_{0}%
^{1}\frac{1}{\left[  1+\frac{\delta\eta^{2}}{1-\delta}\right]  ^{2}}\frac{\pi
}{2}F\left(  \frac{1}{2},-\frac{1}{2},1,(1-\eta^{2})\varepsilon\right)
\mathrm{d}\eta\nonumber\\
& =\frac{ab}{\sqrt[2]{1-\delta}}\int_{0}^{1}\sqrt[2]{1-\varepsilon\eta^{2}%
}\int_{0}^{\pi/2}\frac{1}{\left[  1+\frac{\delta(1-\eta^{2})}{1-\delta}%
\cos^{2}\phi\right]  ^{2}}\mathrm{d}\phi\mathrm{d}\eta\nonumber\\
& =\frac{ab}{\sqrt[2]{1-\delta}}\int_{0}^{1}\sqrt[2]{1-\varepsilon\eta^{2}%
}\frac{\pi}{4}\frac{2+\frac{\delta(1-\eta^{2})}{1-\delta}}{\left[
1+\frac{(\delta(1-\eta^{2}))}{1-\delta}\right]  ^{3/2}}\mathrm{d}\eta\implies
\end{align}
the total surface area is given by%
\begin{align}
\mathcal{A}^{ellipsoid}  & \mathcal{=}\frac{8ab}{\sqrt[2]{1-\delta}}%
\Biggl\{%
\int_{0}^{1}\sqrt[2]{1-\varepsilon\eta^{2}}\frac{\pi}{2}\frac{1}{\left\{
1+\left[  \frac{\delta(1-\eta^{2})}{1-\delta}\right]  \right\}  ^{3/2}%
}\mathrm{d}\eta\nonumber\\
& +\int_{0}^{1}-\frac{\pi}{4}\frac{-\delta(1-\eta^{2})}{1-\delta}%
\frac{\sqrt[2]{1-\varepsilon\eta^{2}}}{\left\{  1+\left[  \frac{\delta
(1-\eta^{2})}{1-\delta}\right]  \right\}  ^{3/2}}\mathrm{d}\eta%
\Biggr\}%
,
\end{align}
while using
\begin{equation}
1+\left[  \frac{\delta(1-\eta^{2})}{1-\delta}\right]  =\frac{1-\delta
+\delta(1-\eta^{2})}{1-\delta}=\frac{1}{1-\delta}(1-\delta\eta^{2}),
\end{equation}
we obtain$:$%
\begin{align}
\mathcal{A}^{ellipsoid}  & =\frac{8ab}{\sqrt[2]{1-\delta}}%
\Biggl\{%
\frac{\pi}{4}\frac{\Gamma(1/2)\Gamma(1)}{\Gamma(3/2)}\frac{1}{[1/(1-\delta
)]^{3/2}}F_{1}\left(  \frac{1}{2},\mathbf{\beta}_{\epsilon},\frac{3}%
{2},\varepsilon,\delta\right) \nonumber\\
& -\frac{\pi}{4}\frac{(-\delta)}{1-\delta}\frac{1}{[1/(1-\delta)]^{3/2}}%
\frac{1}{2}\frac{\Gamma(1/2)\Gamma(2)}{\Gamma(5/2)}F_{1}\left(  \frac{1}%
{2},\mathbf{\beta}_{\epsilon},\frac{5}{2},\varepsilon,\delta\right)
\Biggr\}%
\nonumber\\
& \label{Area51eLL13KRANIOTIS}%
\end{align}
and we defined the 2-tuple:%
\begin{equation}
\mathbf{\beta}_{\epsilon}:=\left(  -\frac{1}{2},\frac{3}{2}\right)  .
\end{equation}
Equation $(\ref{Area51eLL13KRANIOTIS})$ is our \underline{solution in closed
analytic form } for the surface area of the ellipsoid. We believe it
constitutes the first complete exact analytic solution of the problem, while
equation $(\ref{Area51eLL13KRANIOTIS})$ is of certain mathematical beauty.
Thus, we have proved the theorem$:$

\begin{theorem}
\label{GeorgiosVKraniotis} The surface area of the general ellipsoid in closed
analytic form is given by the equation$:$
\begin{align}
& \fbox{$\mathcal{A}_{scalene}^{ellipsoid}=\frac{8ab}{\sqrt[2]{1-\delta}}%
\Biggl\{%
\frac{\pi}{2}\frac{1}{[1/(1-\delta)]^{3/2}}F_{1}\left(  \frac{1}%
{2},\mathbf{\beta}_{\epsilon},\frac{3}{2},\varepsilon,\delta\right)
-\frac{\pi}{6}\frac{(-\delta)}{1-\delta}\frac{1}{[1/(1-\delta)]^{3/2}}%
F_{1}\left(  \frac{1}{2},\mathbf{\beta}_{\epsilon},\frac{5}{2},\varepsilon
,\delta\right)
\Biggr\}%
$}\Longleftrightarrow\nonumber\\
& \fbox{$\mathcal{A}_{scalene}^{ellipsoid}=4\pi ab\left(  \frac{c^{2}}{a^{2}%
}F_{1}(\frac{1}{2},-\frac{1}{2},\frac{3}{2},\frac{3}{2};\epsilon,\delta
)+\frac{1}{3}\left(  1-\frac{c^{2}}{a^{2}}\right)  F_{1}(\frac{1}{2},-\frac
{1}{2},\frac{3}{2},\frac{5}{2};\epsilon,\delta)\right)  $}%
\end{align}

\end{theorem}

The two-variables function $F_{1}\left(  \alpha,\beta,\beta^{\prime}%
,\gamma,x,y\right)  ,$ admits the following integral representation which is
of vital importance in the proof of the theorem eqn.
$(\ref{Area51eLL13KRANIOTIS}),$ for the surface area of the general ellipsoid
$:$%

\begin{equation}%
{\displaystyle\int\limits_{0}^{1}}
u^{\alpha-1}(1-u)^{\gamma-\alpha-1}(1-ux)^{-\beta}(1-uy)^{-\beta^{\prime}%
}du=\frac{\Gamma(\alpha)\Gamma(\gamma-\alpha)}{\Gamma(\gamma)}F_{1}\left(
\alpha,\beta,\beta^{\prime},\gamma,x,y\right)  \overset{}{}\label{INTREP}%
\end{equation}

We point out that in proving the theorem \ref{GeorgiosVKraniotis} we also
produced the following interesting result:

\begin{theorem}%
\begin{align}
& \int_{0}^{1}\sqrt[2]{1-\delta x_{1}^{2}}\frac{\pi}{2}F\left(  \frac{1}%
{2},-\frac{1}{2},1,\frac{(1-x_{1}^{2})}{(1-\delta x_{1}^{2})}(1-c^{2}%
/b^{2})\right)  \mathrm{d}x_{1}\nonumber\\
& =\frac{\pi}{2}(1-\delta)F_{1}\left(  \frac{1}{2},\mathbf{\beta}_{\epsilon
},\frac{3}{2},\varepsilon,\delta\right)  +\frac{\pi}{6}\delta F_{1}\left(
\frac{1}{2},\mathbf{\beta}_{\epsilon},\frac{5}{2},\varepsilon,\delta\right)  .
\end{align}

\end{theorem}

$%
\shadowbox{$\color{blue}Corollaries\;of\;Theorem\;\ref{GeorgiosVKraniotis}.$}%
$

A few special cases follow$.$ In the case: $a=b\neq c$ the two variables of
the hypergeometric function of Appell that appear in
$(\ref{Area51eLL13KRANIOTIS})$ become equal and consequently $F_{1}$ reduces
to the hypergeometric function of Gau\ss :%
\begin{equation}
F_{1}\left(  \frac{1}{2},-\frac{1}{2},\frac{3}{2},\frac{3}{2},1-\frac{c^{2}%
}{a^{2}},1-\frac{c^{2}}{a^{2}}\right)  =F\left(  \frac{1}{2},-\frac{1}%
{2}+\frac{3}{2},\frac{3}{2},1-\frac{c^{2}}{a^{2}}\right)
\end{equation}
and $(\ref{Area51eLL13KRANIOTIS})$ takes the form$:$%
\begin{align}
\mathcal{A}_{a=b\neq c}^{\varepsilon\lambda\lambda\varepsilon\iota\psi
o\varepsilon\iota\delta\acute{\epsilon}\varsigma}  & =\frac{8a^{2}}%
{\sqrt[2]{1-\delta}}%
\Biggl\{%
\frac{\pi}{4}2\frac{1}{[1/(1-\delta)]^{3/2}}F\left(  \frac{1}{2},1,\frac{3}%
{2},1-\frac{c^{2}}{a^{2}}\right) \nonumber\\
& +-\frac{\pi}{4}\frac{(-\delta)}{1-\delta}\frac{1}{[1/(1-\delta)]^{3/2}}%
\frac{1}{2}\frac{4}{3}F\left(  \frac{1}{2},1,\frac{5}{2},1-\frac{c^{2}}{a^{2}%
}\right)
\Biggr\}%
\nonumber\\
& \label{OblateEllip}%
\end{align}
Equation $(\ref{OblateEllip})$ admits a further simplification$:$

\begin{corollary}
In the special case of an ellipsoid with $a=b>c$ ( oblate spheroid$)$ the
following equation is valid$:$%
\begin{align}
\mathcal{A}_{a=b\neq c}^{\varepsilon\lambda\lambda\varepsilon\iota\psi
o\varepsilon\iota\delta\acute{\epsilon}\varsigma}  & =\frac{8a^{2}}%
{\sqrt[2]{1-\delta}}%
\Biggl\{%
\frac{\pi}{4}2\frac{1}{[1/(1-\delta)]^{3/2}}F\left(  \frac{1}{2},1,\frac{3}%
{2},1-\frac{c^{2}}{a^{2}}\right) \nonumber\\
& +-\frac{\pi}{4}\frac{(-\delta)}{1-\delta}\frac{1}{[1/(1-\delta)]^{3/2}}%
\frac{1}{2}\frac{4}{3}F\left(  \frac{1}{2},1,\frac{5}{2},1-\frac{c^{2}}{a^{2}%
}\right)
\Biggr\}%
\nonumber\\
& =2\pi a^{2}+\pi c^{2}\frac{1}{\sqrt[2]{1-c^{2}/a^{2}}}\log\frac
{1+\sqrt[2]{1-c^{2}/a^{2}}}{1-\sqrt[2]{1-c^{2}/a^{2}}}%
.\nonumber\label{ESPOBLATE}\\
&
\end{align}

\end{corollary}

\begin{proof}
We are going to make use of the formula \cite{Thome} $:$%
\begin{equation}
\log\frac{1+x}{1-x}=2xF\left(  \frac{1}{2},1,\frac{3}{2},x^{2}\right)
,\label{tomea1}%
\end{equation}
and the contiguous equation:%
\begin{equation}
F(\alpha,\beta,\gamma+1,z)=\frac{\gamma}{(\gamma-\alpha)(\gamma-\beta)}\left(
(1-z)\frac{\mathrm{d}}{\mathrm{d}z}+\gamma-\alpha-\beta\right)  F(\alpha
,\beta,\gamma,z).\label{thome2}%
\end{equation}
Equation $(\ref{tomea1})$ yields$:$%
\begin{equation}
\fbox{$F\left(  \frac{1}{2},1,\frac{3}{2},1-\frac{c^{2}}{a^{2}}\right)
=\frac{1}{2\text{ }\sqrt[2]{1-c^{2}/a^{2}}}\log\frac{1+\sqrt[2]{1-c^{2}/a^{2}%
}}{1-\sqrt[2]{1-c^{2}/a^{2}}},\label{GVKTOMAE}$}%
\end{equation}
while $(\ref{thome2})$%
\begin{equation}
\fbox{$F\left(  \frac{1}{2},1,\frac{5}{2},1-\frac{c^{2}}{a^{2}}\right)
=\frac{3}{2}\left\{  \frac{1}{1-c^{2}/a^{2}}-\frac{1}{2}\frac{c^{2}/a^{2}%
}{\left(  1-\frac{c^{2}}{a^{2}}\right)  ^{3/2}}\log\frac{1+\sqrt[2]%
{1-c^{2}/a^{2}}}{1-\sqrt[2]{1-c^{2}/a^{2}}}\right\}  .\label{GVKTOM13}$}%
\end{equation}
Substituting equations $(\ref{GVKTOMAE}),(\ref{GVKTOM13})$ into
$(\ref{OblateEllip})$ the corollary is proved.
\end{proof}

\begin{corollary}
In the special case of an ellipsoid with $a>b=c$ (prolate spheroid$)$ it
holds$:$%
\begin{align}
\mathcal{A}_{a\neq b=c}^{ellipsoid}  & =\frac{8a^{2}}{\sqrt[2]{1-\delta}}%
\Biggl\{%
\frac{\pi}{4}2\frac{1}{[1/(1-\delta)]^{3/2}}F\left(  \frac{1}{2},\frac{3}%
{2},\frac{3}{2},1-\frac{c^{2}}{a^{2}}\right) \nonumber\\
& +-\frac{\pi}{4}\frac{(-\delta)}{1-\delta}\frac{1}{[1/(1-\delta)]^{3/2}}%
\frac{1}{2}\frac{4}{3}F\left(  \frac{1}{2},\frac{3}{2},\frac{5}{2}%
,1-\frac{c^{2}}{a^{2}}\right)
\Biggr\}%
\nonumber\\
& =2ab\pi\left(  \sqrt[2]{\frac{b^{2}}{a^{2}}}+\frac{\sqrt[2]{(1-b^{2}/a^{2}%
)}}{1-b^{2}/a^{2}}\arcsin\sqrt[2]{1-b^{2}/a^{2}}\right)  .\nonumber
\end{align}

\end{corollary}

\begin{proof}
Indeed$,$ for $b=c,$ the first variable of the Appell's functions in the
closed form solution of $(\ref{Area51eLL13KRANIOTIS})$ vanishes and the
generalized hypergeometric functions in discussion reduce as follows$:$%
\begin{align}
F_{1}\left(  \frac{1}{2},\mathbf{\beta}_{\varepsilon},\frac{3}{2}%
,0,1-\frac{c^{2}}{a^{2}}\right)   & =F\left(  \frac{1}{2},\frac{3}{2},\frac
{3}{2},1-\frac{c^{2}}{a^{2}}\right)  ,\\
F_{1}\left(  \frac{1}{2},\mathbf{\beta}_{\varepsilon},\frac{5}{2}%
,0,1-\frac{c^{2}}{a^{2}}\right)   & =F\left(  \frac{1}{2},\frac{3}{2},\frac
{5}{2},1-\frac{c^{2}}{a^{2}}\right)
\end{align}
We now apply the formulae $:$%
\begin{equation}
F\left(  \frac{1}{2},\frac{1}{2},\frac{3}{2},x\right)  =\frac{1}{\sqrt[2]{x}%
}\arcsin\sqrt[2]{x},\label{EIDIKOS}%
\end{equation}%
\begin{align}
F(\alpha,\beta+1,\gamma+1,x)-F(\alpha,\beta,\gamma,x)  & =\left(  \frac
{\gamma-\beta}{\beta\gamma}\right)  x\frac{\mathrm{d}}{\mathrm{d}x}%
F(\alpha,\beta,\gamma+1,x),\nonumber\\
& \\
F(\alpha,\beta+1,\gamma,x)  & =F(\alpha,\beta,\gamma,x)+\frac{x}{\beta}%
\frac{\mathrm{d}}{\mathrm{d}x}F(\alpha,\beta,\gamma,x).\nonumber\\
&
\end{align}
We thus end up with the equations$:$%
\begin{align}
F\left(  \frac{1}{2},\frac{3}{2},\frac{3}{2},x\right)   & =\frac{1}%
{\sqrt[2]{1-x}},\\
F\left(  \frac{1}{2},\frac{3}{2},\frac{5}{2},x\right)   & =\frac{3}{2}%
\frac{\arcsin\sqrt[2]{x}}{x^{3/2}}-\frac{3}{2}\frac{\sqrt[2]{1-x}}%
{x},\label{Contiguous1}%
\end{align}
and the corollary is proved$.$
\end{proof}

\begin{corollary}
\bigskip\bigskip In the special case of an ellipsoid with $a=c<b$ $($prolate
spheroid$)$
\begin{align}
\mathcal{A}_{a=c<b}^{ellipsoid}  & =\frac{8ab}{1}\frac{\pi}{2}F_{1}\left(
\frac{1}{2},-\frac{1}{2},\frac{3}{2},\frac{3}{2},1-\frac{c^{2}}{b^{2}%
},0\right) \nonumber\\
& =4ab\pi F\left(  \frac{1}{2},-\frac{1}{2},\frac{3}{2},1-\frac{c^{2}}{b^{2}%
}\right) \nonumber\\
& =2ab\pi\left\{  \frac{c}{b}+\frac{\arcsin\sqrt[2]{1-c^{2}/b^{2}}}%
{\sqrt[2]{1-c^{2}/b^{2}}}\right\}  .
\end{align}

\end{corollary}

\begin{proof}
Here we make use of the formulae$:$%
\begin{equation}
F_{1}(\alpha,\beta,\beta^{\prime},\gamma,x,0)=F(\alpha,\beta,\gamma,x),
\end{equation}%
\begin{equation}
F(\alpha,\beta-1,\gamma,x)=\frac{1}{\gamma-\beta}\left[  z(1-z)\frac
{\mathrm{d}}{\mathrm{d}z}-\alpha z-\beta+\gamma\right]  F(\alpha,\beta
,\gamma,x)
\end{equation}
and $(\ref{EIDIKOS}).$
\end{proof}

\begin{corollary}
\bigskip\bigskip In the particular case of the sphere $\ a=b=c$
$(\ref{Area51eLL13KRANIOTIS})$ reduces to the known result $\mathcal{A=}4\pi
a^{2}. $
\end{corollary}

Let us give some examples of $(\ref{Area51eLL13KRANIOTIS}).$ For
$a=2,b=1,c=0.25$ using $(\ref{Area51eLL13KRANIOTIS})$ $\ $we compute:
$\mathcal{A}^{ellipsoid}=13.6992108087.$ For $a=1,b=1,c=0.5,$ we calculate
$\mathcal{A}^{ellipsoid}=8.6718827033,$ and \ for $a=1,b=0.8,c=0.625,$ we
compute: $\mathcal{A}^{ellipsoid}=8.1516189229.$

\bigskip

\section{Capacitance of a conducting ellipsoid.}

For a conducting 3-dimensional ellipsoid $\mathbf{E}$ given by equation
(\ref{ellipsoidgeniko}) its electrostatic capacitance is determined by solving
Laplace's equation:%
\begin{equation}
\nabla^{2}\Phi=\frac{\partial^{2}\Phi}{\partial x^{2}}+\frac{\partial^{2}\Phi
}{\partial y^{2}}+\frac{\partial^{2}\Phi}{\partial z^{2}}=0,\label{Laplace}%
\end{equation}
outside $\mathbf{E,}$ subject to $\Phi=1$ on the surface and $\Phi=0$ at
$\infty$. The function $\Phi(\vec{x})$ is the (equilibrium) electrostatic
potential of $\mathbf{E.}$

The computation is facilitated by ellipsoidal coordinates, a trick first
introduced by Jacobi \cite{JACOBI}. A point $\vec{x}$ outside $\mathbf{E}$
determines the parameter $0\leq r=r(\vec{x})<\infty$ by means of%
\begin{equation}
\frac{x^{2}}{a^{2}+r}+\frac{y^{2}}{b^{2}+r}+\frac{z^{2}}{c^{2}+r}=1.
\end{equation}
Now we look at:%
\begin{equation}
\Phi(\vec{x})=\frac{1}{2}\int_{r(\vec{x})}^{\infty}\left[  (a^{2}%
+r)(b^{2}+r)(c^{2}+r)\right]  ^{-1/2}\mathrm{d}r
\end{equation}
outside $\mathbf{E.}$ It satisfies (\ref{Laplace}) and the prescribed
conditions, therefore the capacitance is given by \cite{GJTEE},\cite{Polya}:%
\begin{equation}
C^{-1}(\mathbf{E)=}\frac{1}{2}\int_{0}^{\infty}\frac{\mathrm{d}r}%
{\sqrt[2]{(a^{2}+r)(b^{2}+r)(c^{2}+r)}}\label{ellipticinteC}%
\end{equation}
We will now show the following theorem in which we compute exactly the
elliptic integral in $(\ref{ellipticinteC})$:

\begin{theorem}
\label{XWRITIKOTITA}The closed form solution for the capacitance of a
3-dimensional conducting ellipsoid is given by the expression:%
\begin{equation}
\shadowbox{$\displaystyle C({\bf E})=\frac{2a\Gamma(3/2)}{\Gamma
(1/2)\Gamma(1)}\frac{1}{F_1\left(\frac{1}{2},\frac{1}{2},\frac{1}{2},\frac
{3}{2},x,y\right)},$}
\label{CapaciGVKellipsoid}
\end{equation}
where $x:=1-\frac{b^{2}}{a^{2}},y:=1-\frac{c^{2}}{a^{2}}$ and we organize the
axes so that: $a>b>c>0.$
\end{theorem}

Equivalently, by substituting the values for the gamma function into $($%
\ref{CapaciGVKellipsoid}%
$)$: $\Gamma(3/2)=\frac{\sqrt[2]{\pi}}{2},$ $\Gamma(1/2)=\sqrt[2]{\pi},$ we derive:%

\begin{equation}
C(\mathbf{E)=}\frac{a}{F_{1}\left(  \frac{1}{2},\frac{1}{2},\frac{1}{2}%
,\frac{3}{2},1-\frac{b^{2}}{a^{2}},1-\frac{c^{2}}{a^{2}}\right)
}.\label{capacizeppelin}%
\end{equation}

For the proof of theorem \ref{XWRITIKOTITA} see Appendix.

We also derive the following corollaries of Theorem \ref{XWRITIKOTITA}:

\begin{corollary}
For the special case $a=b>c$ the capacitance of the spheroid is given by%
\begin{equation}
C_{a=b>c}=\frac{2a\Gamma(3/2)}{\Gamma(1/2)}\frac{\sqrt[2]{1-c^{2}/a^{2}}%
}{\arcsin\sqrt[2]{1-c^{2}/a^{2}}}.
\end{equation}

\end{corollary}

\begin{proof}
For $a=b>c,$ Eq. $(\ref{capacizeppelin})$ reduces to:%
\begin{align}
C  & =\frac{2a\Gamma(3/2)}{\Gamma(1/2)}\frac{1}{F_{1}(\frac{1}{2},\frac{1}%
{2},\frac{1}{2},\frac{3}{2},0,1-\frac{c^{2}}{a^{2}})}\nonumber\\
& =a\frac{1}{F\left(  \frac{1}{2},\frac{1}{2},\frac{3}{2},1-\frac{c^{2}}%
{a^{2}}\right)  }\overset{(\ref{EIDIKOS})}{=}a\frac{\sqrt[2]{1-c^{2}/a^{2}}%
}{\arcsin\sqrt[2]{1-c^{2}/a^{2}}}.
\end{align}

\end{proof}

\begin{corollary}
For $b=c<a,$ the case of a prolate spheroid we derive%
\begin{equation}
C_{b=c<a}=a\frac{2\sqrt[2]{1-c^{2}/a^{2}}}{\log\frac{1+\sqrt[2]{1-c^{2}/a^{2}%
}}{1-\sqrt[2]{1-c^{2}/a^{2}}}}%
\end{equation}

\end{corollary}

\begin{proof}
For $c=b<a,$ Eq. $(\ref{capacizeppelin})$ reduces to:%
\begin{align}
C  & =\frac{2a\Gamma(3/2)}{\Gamma(1/2)}\frac{1}{F_{1}\left(  \frac{1}{2}%
,\frac{1}{2},\frac{1}{2},\frac{3}{2},1-\frac{c^{2}}{a^{2}},1-\frac{c^{2}%
}{a^{2}}\right)  }\nonumber\\
& =a\frac{1}{F\left(  \frac{1}{2},1,\frac{3}{2},1-\frac{c^{2}}{a^{2}}\right)
}\overset{(\ref{tomea1})}{=}a\frac{2\sqrt[2]{1-c^{2}/a^{2}}}{\log
\frac{1+\sqrt[2]{1-c^{2}/a^{2}}}{1-\sqrt[2]{1-c^{2}/a^{2}}}}%
\end{align}

\end{proof}

\begin{corollary}
For a conducting sphere, $a=b=c,$ and equation $($%
\ref{CapaciGVKellipsoid}%
$)$ reduces to:%
\begin{equation}
C_{sphere}=a.
\end{equation}

\end{corollary}

\begin{proof}
For the case of the sphere the first hypergeometric function of Appell in $($%
\ref{CapaciGVKellipsoid}%
$)$ takes the value $1$.
\end{proof}

\begin{corollary}
For the case of an elliptic disk, $a>b>c=0,$ the capacitance is given in
closed analytic form by:%
\begin{equation}
C_{elliptic\text{ }disk}=\frac{a}{\frac{\pi}{2}F\left(  \frac{1}{2},\frac
{1}{2},1,1-\frac{b^{2}}{a^{2}}\right)  }.
\end{equation}

\end{corollary}

\begin{proof}
Indeed, in the case of an elliptic disk equation $($%
\ref{CapaciGVKellipsoid}%
$)$ $\ $reduces as follows:%
\begin{align}
C  & =\frac{a}{F_{1}\left(  \frac{1}{2},\frac{1}{2},\frac{1}{2},\frac{3}%
{2},1-\frac{b^{2}}{a^{2}},1\right)  }\nonumber\\
& =\frac{a}{\frac{\Gamma(3/2)\Gamma(3/2-1/2-1/2)}{\Gamma(3/2-1/2)\Gamma
(3/2-1/2)}F\left(  \frac{1}{2},\frac{1}{2},\frac{3}{2}-\frac{1}{2}%
,1-\frac{b^{2}}{a^{2}}\right)  }\nonumber\\
& =\frac{a}{\frac{\pi}{2}F\left(  \frac{1}{2},\frac{1}{2},1,1-\frac{b^{2}%
}{a^{2}}\right)  }%
\end{align}

\end{proof}

We now apply our closed form analytic solution $(\ref{capacizeppelin})$ for
computing the capacitance of some ellipsoids. Our results are presented in
Table \ref{CAPELIIP}.%

\begin{table}[tbp] \centering
\begin{tabular}
[c]{|l|l|l|}\hline
$a=5,b=2,c=1$ & $a=10,b=5,c=4$ & $a=10,b=8,c=5$\\\hline
$C/a=0.50822148949$ & $C/a=0.621383016235$ & $C/a=0.76121621804$\\
&  & \\\hline
\end{tabular}
\caption{Capacitance of some ellipsoids computed from
(\ref{CapaciGVKellipsoid})}\label{CAPELIIP}%
\end{table}%

\bigskip

In Figure
\ref{GVKCAPELLI}
, we plot the capacitance $C(\mathbf{E})/a$ using Eqn.$(\ref{capacizeppelin}%
),$\ of a conducting ellipsoid immersed in $%
\mathbb{R}
^{3}$ versus the ratio $c/a$ of the axes for various values of the ratio
$b/a.$

\section{Demagnetizing factors of a magnetized ellipsoid}

The first attempts of calculating the demagnetizing factors \ of the ellipsoid
were made in \cite{Kellogg},\cite{OSBORN}. However, the authors in
\cite{Kellogg},\cite{OSBORN}, only derived expressions in terms of formal
integrals. We now derive the \textit{first } \textbf{closed form analytic
solution} \ for the demagnetizing factors of the magnetized ellipsoid in terms
of Appell's first hypergeometric function $F_{1}$.

The potential in the interior of the magnetized ellipsoid is \ a quadratic
function of $x,y,$ and $z$ \cite{Kellogg}:%
\begin{equation}
\Phi_{int}=-Lx^{2}-My^{2}-Nz^{2}+W,
\end{equation}
where the demagnetizing factors are defined by the integrals:%
\begin{align}
L  & =\pi a_{1}a_{2}a_{3}\kappa\int_{0}^{\infty}\frac{\mathrm{d}u}{(a_{1}%
^{2}+u)\sqrt{(a_{1}^{2}+u)(a_{2}^{2}+u)(a_{3}^{2}+u)}},\label{LDEM1}\\
M  & =\pi a_{1}a_{2}a_{3}\kappa\int_{0}^{\infty}\frac{\mathrm{d}u}{(a_{2}%
^{2}+u)\sqrt{(a_{1}^{2}+u)(a_{2}^{2}+u)(a_{3}^{2}+u)}},\label{mdem2}\\
N  & =\pi a_{1}a_{2}a_{3}\kappa\int_{0}^{\infty}\frac{\mathrm{d}u}{(a_{3}%
^{2}+u)\sqrt{(a_{1}^{2}+u)(a_{2}^{2}+u)(a_{3}^{2}+u)}},\label{ndem3}%
\end{align}
where we have the correspondence:%
\begin{align}
a_{1}  & =a,a_{2}=b,a_{3}=c,\\
a  & >b>c.
\end{align}
and%
\begin{equation}
W=\pi a_{1}a_{2}a_{3}\kappa\int_{0}^{\infty}\frac{\mathrm{d}u}{\sqrt
{(a_{1}^{2}+u)(a_{2}^{2}+u)(a_{3}^{2}+u)}}.
\end{equation}
From Poisson's differential equation%
\begin{equation}
\nabla^{2}\Phi_{int}=-2(L+M+N)=-4\pi\kappa,
\end{equation}
we derive the relationship that the demagnetizing factors satisfy:%
\begin{equation}
L+M+N=2\pi\kappa.
\end{equation}

\begin{theorem}
\label{magnetization}The closed form solution of the demagnetizing factors
$L,M,N$, of the magnetized scalene ellipsoid is the following:%
\begin{align}
L  & =\frac{\pi abc}{a^{3}}\kappa\frac{\Gamma\left(  \frac{3}{2}\right)
\Gamma(1)}{\Gamma\left(  \frac{5}{2}\right)  }F_{1}\left(  \frac{3}{2}%
,\frac{1}{2},\frac{1}{2},\frac{5}{2},1-\frac{b^{2}}{a^{2}},1-\frac{c^{2}%
}{a^{2}}\right)  ,\label{magnitisena}\\
M  & =\frac{\pi abc}{a^{3}}\kappa\frac{\Gamma\left(  \frac{3}{2}\right)
\Gamma(1)}{\Gamma\left(  \frac{5}{2}\right)  }F_{1}\left(  \frac{3}{2}%
,\frac{3}{2},\frac{1}{2},\frac{5}{2},1-\frac{b^{2}}{a^{2}},1-\frac{c^{2}%
}{a^{2}}\right)  ,\label{magnitis2}\\
N  & =\frac{\pi abc}{a^{3}}\kappa\frac{\Gamma\left(  \frac{3}{2}\right)
\Gamma(1)}{\Gamma\left(  \frac{5}{2}\right)  }F_{1}\left(  \frac{3}{2}%
,\frac{1}{2},\frac{3}{2},\frac{5}{2},1-\frac{b^{2}}{a^{2}},1-\frac{c^{2}%
}{a^{2}}\right)  .\label{magnitis3}%
\end{align}

\end{theorem}

\begin{proof}
The demagnetizing factors are defined by the integrals eqns(\ref{LDEM1}%
)-(\ref{ndem3}). We now compute these integrals in closed analytic form. We
apply the transformation:%
\begin{equation}
1+\frac{u}{a_{1}^{2}}=\frac{1}{x^{2}}\Rightarrow\mathrm{d}u=-\frac{2a_{1}^{2}%
}{x^{3}}\mathrm{d}x.\label{Metasximatismos}%
\end{equation}
Consequently:%
\begin{equation}
L=\frac{2\pi a_{1}a_{2}a_{3}\kappa}{a_{1}^{3}}\int_{0}^{1}\frac{\mathrm{d}%
x\text{ }x^{2}}{\sqrt[2]{(1-\mu_{2}x^{2})(1-\mu_{3}x^{2})}},
\end{equation}
where we defined the moduli:%
\begin{equation}
\mu_{2}:=1-\frac{a_{2}^{2}}{a_{1}^{2}},\text{ }\mu_{3}:=1-\frac{a_{3}^{2}%
}{a_{1}^{2}}.
\end{equation}
We now set:%
\begin{equation}
x^{2}=\xi\Rightarrow\mathrm{d}x=\frac{\mathrm{d}\xi}{2\sqrt[2]{\xi}}.
\end{equation}
Thus%
\begin{align}
L  & =\frac{\pi a_{1}a_{2}a_{3}\kappa}{a_{1}^{3}}\int_{0}^{1}\frac
{\mathrm{d}\xi\xi^{1/2}}{\sqrt[2]{(1-\mu_{2}\xi)(1-\mu_{3}\xi)}}\nonumber\\
& \overset{(\ref{INTREP})}{=}\frac{\pi abc\kappa}{a^{3}}\frac{\Gamma
(3/2)\Gamma(1)}{\Gamma(5/2)}F_{1}\left(  \frac{3}{2},\frac{1}{2},\frac{1}%
{2},\frac{5}{2},1-\frac{b^{2}}{a^{2}},1-\frac{c^{2}}{a^{2}}\right)
\end{align}
In a similar way, repeating the previous transformations, we compute
analytically the two other demagnetizing factors $M,N.$ For instance:%
\begin{align}
N  & =\frac{\pi a_{1}a_{2}a_{3}\kappa}{a_{1}^{3}}\int_{0}^{1}\frac
{\mathrm{d}\xi\xi^{1/2}}{(1-\mu_{3}\xi)\sqrt[2]{(1-\mu_{2}\xi)(1-\mu_{3}\xi)}%
}\nonumber\\
& =\frac{\pi bc\kappa}{a^{2}}\frac{\Gamma(3/2)\Gamma(1)}{\Gamma(5/2)}%
F_{1}\left(  \frac{3}{2},\frac{1}{2},\frac{3}{2},\frac{5}{2},1-\frac{b^{2}%
}{a^{2}},1-\frac{c^{2}}{a^{2}}\right)  .
\end{align}

\end{proof}

We now study special cases for the demagnetizing factors of the magnetized ellipsoid.

\begin{corollary}
In the special case $a=b>c$%
\begin{align}
L  & =\pi\frac{c}{a}\kappa\frac{\Gamma\left(  \frac{3}{2}\right)  \Gamma
(1)}{\Gamma\left(  \frac{5}{2}\right)  }F_{1}\left(  \frac{3}{2},\frac{1}%
{2},\frac{1}{2},\frac{5}{2},0,1-\frac{c^{2}}{a^{2}}\right) \nonumber\\
& =\pi\frac{c}{a}\kappa\frac{\Gamma\left(  \frac{3}{2}\right)  \Gamma
(1)}{\Gamma\left(  \frac{5}{2}\right)  }F\left(  \frac{3}{2},\frac{1}{2}%
,\frac{5}{2},1-\frac{c^{2}}{a^{2}}\right)  =\frac{\pi\kappa m^{2}}%
{(m^{2}-1)^{3/2}}\arcsin\frac{\sqrt[2]{m^{2}-1}}{m}-\frac{\pi\kappa}{m^{2}%
-1},\nonumber\\
& \\
M  & =\pi\frac{c}{a}\kappa\frac{\Gamma\left(  \frac{3}{2}\right)  \Gamma
(1)}{\Gamma\left(  \frac{5}{2}\right)  }F\left(  \frac{3}{2},\frac{1}{2}%
,\frac{5}{2},1-\frac{c^{2}}{a^{2}}\right)  =\frac{\pi\kappa m^{2}}%
{(m^{2}-1)^{3/2}}\arcsin\frac{\sqrt[2]{m^{2}-1}}{m}-\frac{\pi\kappa}{m^{2}%
-1},\nonumber\\
& \\
N  & =\pi\frac{c}{a}\kappa\frac{\Gamma\left(  \frac{3}{2}\right)  \Gamma
(1)}{\Gamma\left(  \frac{5}{2}\right)  }F_{1}\left(  \frac{3}{2},\frac{1}%
{2},\frac{3}{2},\frac{5}{2},0,1-\frac{c^{2}}{a^{2}}\right) \nonumber\\
& =\pi\frac{c}{a}\kappa\frac{\Gamma\left(  \frac{3}{2}\right)  \Gamma
(1)}{\Gamma\left(  \frac{5}{2}\right)  }F\left(  \frac{3}{2},\frac{3}{2}%
,\frac{5}{2},1-\frac{c^{2}}{a^{2}}\right)  =\frac{2\pi\kappa m^{2}}{m^{2}%
-1}\left\{  \frac{-\arcsin\frac{\sqrt[2]{m^{2}-1}}{m}}{\sqrt[2]{m^{2}-1}%
}+1\right\} \label{NLIMITAEQB}%
\end{align}
where $m:=a/c.$
\end{corollary}

\begin{proof}
For $a=b>c,$%
\begin{equation}
L=M=\pi\frac{c}{a}\kappa\frac{\Gamma\left(  \frac{3}{2}\right)  \Gamma
(1)}{\Gamma\left(  \frac{5}{2}\right)  }F\left(  \frac{3}{2},\frac{1}{2}%
,\frac{5}{2},1-\frac{c^{2}}{a^{2}}\right)  .
\end{equation}
Using the contiguous relations
\begin{align}
z\frac{\mathrm{d}}{\mathrm{d}z}F(\alpha,\beta,\gamma,z)  & =\alpha\left[
F(\alpha+1,\beta,\gamma,z)-F(\alpha,\beta,\gamma,z)\right]  ,\label{contgauss}%
\\
& \nonumber
\end{align}
and Eq.(\ref{thome2})%
\begin{equation}
\frac{1}{2}F\left(  \frac{3}{2},\frac{1}{2},\frac{5}{2},\overset
{:=z}{\overbrace{1-\frac{c^{2}}{a^{2}}}}\right)  =\frac{1}{2}F\left(  \frac
{1}{2},\frac{1}{2},\frac{5}{2},z\right)  +z\frac{\mathrm{d}}{\mathrm{d}%
z}F\left(  \frac{1}{2},\frac{1}{2},\frac{5}{2},z\right)  ,
\end{equation}
with
\begin{equation}
F\left(  \frac{1}{2},\frac{1}{2},\frac{5}{2},z\right)  =\frac{3}{2}\left[
(1-z)\left\{  -\frac{1}{2z^{3/2}}\arcsin\sqrt[2]{z}+\frac{1}{2z}\frac
{1}{\sqrt[2]{1-z}}\right\}  +\frac{1}{2}\frac{\arcsin\sqrt[2]{z}}{\sqrt[2]{z}%
}\right]  .
\end{equation}
Consequently:%
\begin{equation}
L=M=\frac{\pi\kappa m^{2}}{(m^{2}-1)^{3/2}}\arcsin\frac{\sqrt[2]{m^{2}-1}}%
{m}-\frac{\pi\kappa}{m^{2}-1},
\end{equation}
where $z=1-\frac{c^{2}}{a^{2}}:=1-\frac{1}{m^{2}}$. On the other hand, using the
contiguous relation (\ref{contgauss}) for the calculation of the
$N$-demagnetizing factor, yields:%
\begin{align}
z\frac{\mathrm{d}}{\mathrm{d}z}F(\alpha,\beta,\gamma,z)  & =\alpha\left[
F(\alpha+1,\beta,\gamma,z)-F(\alpha,\beta,\gamma,z)\right]  \Rightarrow\\
\frac{1}{2}F\left(  \frac{3}{2},\frac{3}{2},\frac{5}{2},z\right)   & =\frac
{1}{2}F\left(  \frac{1}{2},\frac{3}{2},\frac{5}{2},z\right)  +z\frac
{\mathrm{d}}{\mathrm{d}z}F\left(  \frac{1}{2},\frac{3}{2},\frac{5}%
{2},z\right)  ,
\end{align}
where the Gau\ss \ function $F\left(  \frac{1}{2},\frac{3}{2},\frac{5}%
{2},z\right)  $ is given by Equation (\ref{Contiguous1}). Consequently,
\begin{equation}
\frac{1}{2}F\left(  \frac{3}{2},\frac{3}{2},\frac{5}{2},z\right)  =-\frac
{6}{4}\frac{\arcsin\sqrt[2]{z}}{z^{3/2}}+\frac{3}{4}\frac{\sqrt[2]{1-z}}%
{z}+\frac{3}{4}\frac{1}{\sqrt[2]{1-z}}\left(  \frac{1}{z}+1\right)  ,
\end{equation}
and $N$ is given by Equation (\ref{NLIMITAEQB}).
\end{proof}

\begin{corollary}
In the special case $b=c<a$
\begin{align}
L  & =\frac{\pi c^{2}}{a^{2}}\frac{\kappa\Gamma\left(  \frac{3}{2}\right)
\Gamma(1)}{\Gamma\left(  \frac{5}{2}\right)  }F\left(  \frac{3}{2},1,\frac
{5}{2},1-\frac{c^{2}}{a^{2}}\right) \nonumber\\
& =\frac{2\pi\kappa}{m^{2}-1}\left\{  \frac{m}{2}\log\frac{m+\sqrt[2]{m^{2}%
-1}}{m-\sqrt[2]{m^{2}-1}}-1\right\}  ,m:=a/c,\\
N  & =\frac{\pi c^{2}}{a^{2}}\frac{\kappa\Gamma\left(  \frac{3}{2}\right)
\Gamma(1)}{\Gamma\left(  \frac{5}{2}\right)  }F\left(  \frac{3}{2},2,\frac
{5}{2},1-\frac{c^{2}}{a^{2}}\right) \nonumber\\
& =\kappa\pi\frac{m}{m^{2}-1}\left[  -\frac{1}{2}\frac{1}{\sqrt[2]{m^{2}-1}%
}\log\frac{m+\sqrt[2]{m^{2}-1}}{m-\sqrt[2]{m^{2}-1}}+m\right]
\end{align}

\end{corollary}

\begin{proof}
For $b=c<a$%
\begin{align}
N  & =\frac{\pi ac^{2}}{a^{3}}\kappa\frac{\Gamma\left(  \frac{3}{2}\right)
\Gamma(1)}{\Gamma\left(  \frac{5}{2}\right)  }F_{1}\left(  \frac{3}{2}%
,\frac{1}{2},\frac{3}{2},\frac{5}{2},1-\frac{c^{2}}{a^{2}},1-\frac{c^{2}%
}{a^{2}}\right) \nonumber\\
& =\frac{\pi ac^{2}}{a^{3}}\kappa\frac{\Gamma\left(  \frac{3}{2}\right)
\Gamma(1)}{\Gamma\left(  \frac{5}{2}\right)  }F\left(  \frac{3}{2},2,\frac
{5}{2},\overset{z:=}{\overbrace{1-\frac{c^{2}}{a^{2}}}}\right)  =\frac{\pi
c^{2}}{a^{2}}\frac{\kappa\Gamma\left(  \frac{3}{2}\right)  \Gamma(1)}%
{\Gamma\left(  \frac{5}{2}\right)  }3\frac{\mathrm{d}}{\mathrm{d}z}F\left(
\frac{1}{2},1,\frac{3}{2},z\right) \nonumber\\
& \overset{(\ref{tomea1})}{=}\frac{\pi c^{2}}{a^{2}}\frac{\kappa\Gamma\left(
\frac{3}{2}\right)  \Gamma(1)}{\Gamma\left(  \frac{5}{2}\right)  }\frac{3}%
{2}\left\{  -\frac{1}{2}\frac{1}{z^{3/2}}\log\frac{1+\sqrt[2]{z}}%
{1-\sqrt[2]{z}}+\frac{1}{z(1-z)}\right\} \nonumber\\
& \overset{z=1-1/m^{2}}{=}\kappa\pi\frac{m}{m^{2}-1}\left[  -\frac{1}{2}%
\frac{1}{\sqrt[2]{m^{2}-1}}\log\frac{m+\sqrt[2]{m^{2}-1}}{m-\sqrt[2]{m^{2}-1}%
}+m\right]
\end{align}

\end{proof}

\bigskip

\bigskip

\begin{corollary}
\label{magsphere}For a magnetized sphere$,$ $a=b=c,$ and the demagnezing
factors are all equal to:%
\begin{equation}
L=M=N=\frac{2\pi\kappa}{3}.
\end{equation}

\end{corollary}

\begin{proof}
In this case, the generalized hypergeometric functions of Appell that appear
in the exact solutions for the demagnetizing factors, eqns$(\ref{magnitisena}%
)-(\ref{magnitis3}),$ take the value $1.$
\end{proof}

\bigskip

\bigskip

\bigskip

\bigskip

\bigskip

\bigskip We now apply our closed form analytic solutions for the demagnetizing
factors, Eqns. $(\ref{magnitisena})-(\ref{magnitis3}),$ in order to compute
these factors for various ellipsoids. \ Our results are summarized in Tables
\ref{LMNEllipsoid} $\footnote{Osborn \cite{OSBORN} gave the value:
$M/4\pi=0.306$ for the particular ellipsoid, second column of Table
\ref{LMNEllipsoid}.},\ref{GVKGOSBORN}$.%

\begin{table}[tbp] \centering
\begin{tabular}
[c]{|l|l|l|}\hline
$a=3,b=2,c=1$ & $a=4,b=3,c=2$ & $a=10,b=3,c=2$\\\hline
$\frac{L}{4\pi}=0.156300698829271$ & $\frac{L}{4\pi}=0.211265605319304$ &
$\frac{L}{4\pi}=0.0725156494555862$\\
$\frac{M}{4\pi}=0.267154040262005$ & $\frac{M}{4\pi}=0.305006257867421$ &
$\frac{M}{4\pi}=0.366221770806668$\\
$\frac{N}{4\pi}=0.576545260908724$ & $\frac{N}{4\pi}=0.483728136813275$ &
$\frac{N}{4\pi}=0.561262579737746$\\\hline
\end{tabular}
\caption{The demagnetizing factors $L/4\pi,M/4\pi,N/4\pi$ as computed from
the equations of Theorem \ref{magnetization}, for various ellipsoids. }\label{LMNEllipsoid}%
\end{table}%
%

\begin{table}[tbp] \centering
\begin{tabular}
[c]{|l|l|l|l|l|l|}\hline
& $\frac{c}{a}$ & $\frac{b}{a}$ & $\frac{L}{4\pi}$ & $\frac{M}{4\pi}$ &
$\frac{N}{4\pi}$\\\hline
1.$%
\begin{array}
[c]{c}%
\text{Theorem \ref{magnetization} }\\
\text{Numerical Osborn}%
\end{array}
$ & $0.017452$ & $0.99620$ & $%
\begin{array}
[c]{c}%
0.0133953\\
\end{array}
$ & $%
\begin{array}
[c]{c}%
0.0134714\\
0.013471
\end{array}
$ & $%
\begin{array}
[c]{c}%
0.973133\\
\end{array}
$\\
2.$%
\begin{array}
[c]{c}%
\text{Theorem \ref{magnetization} }\\
\text{Numerical Osborn}%
\end{array}
$ & $0.087156$ & $0.984920$ & $%
\begin{array}
[c]{c}%
0.0613072\\
0.06108
\end{array}
$ & $%
\begin{array}
[c]{c}%
0.062672\\
0.06281
\end{array}
$ & $%
\begin{array}
[c]{c}%
0.876021\\
0.87611
\end{array}
$\\
3.$%
\begin{array}
[c]{c}%
\text{Theorem \ref{magnetization} }\\
\text{Numerical Osborn}%
\end{array}
$ & $0.5$ & $0.98863$ & $%
\begin{array}
[c]{c}%
0.235445\\
0.23555
\end{array}
$ & $%
\begin{array}
[c]{c}%
0.238955\\
0.23885
\end{array}
$ & $%
\begin{array}
[c]{c}%
0.5256\\
0.5256
\end{array}
$\\\hline
4.$%
\begin{array}
[c]{c}%
\text{Theorem \ref{magnetization} }\\
\text{Numerical Osborn}%
\end{array}
$ & $0.087156$ & $1$ & $%
\begin{array}
[c]{c}%
0.0615658\\
0.06154
\end{array}
$ & $%
\begin{array}
[c]{c}%
0.0615658\\
0.06154
\end{array}
$ & $%
\begin{array}
[c]{c}%
0.876868\\
0.87692
\end{array}
$\\\hline
\end{tabular}
\caption{The demagnetizing factors for various values of the ratios
$c/a,b/a$ and a comparison with the numerical results of reference
\cite{OSBORN}. Empty entries in the table means that the author of
\cite{OSBORN} did not provide such values.}\label{GVKGOSBORN}%
\end{table}%

We also plot the demagnetizing factors $L/4\pi,M/4\pi,N/4\pi$ of the
magnetized ellipsoid versus the ratio $c/a,$ for various values of the ratio
$b/a.$ Our results are displayed in Figures \ref{LcdemFac},\ref{MdemFac1}%
,\ref{NdemFac}.

\bigskip%

\begin{figure}
[ptb]
\begin{center}
\includegraphics[
height=4.2134in,
width=6.3815in
]%
{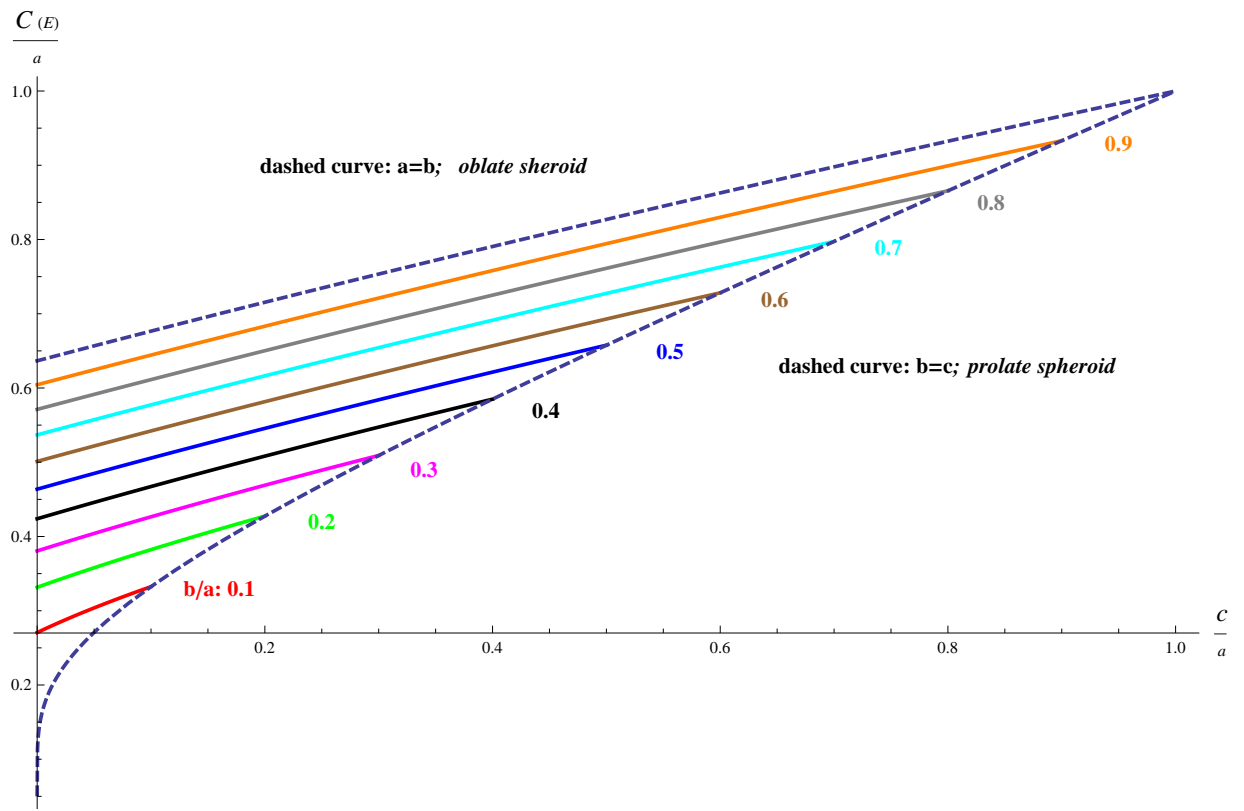}%
\caption{The capacitance $C(\mathbf{E})$ of a conducting ellipsoid immersed in
$\mathbb{R} ^{3}$ versus the ratio $c/a$ of the axes for various values of the
ratio $b/a.$ }%
\label{GVKCAPELLI}%
\end{center}
\end{figure}
%

\begin{figure}
[ptb]
\begin{center}
\includegraphics[
height=3.5907in,
width=5.4483in
]%
{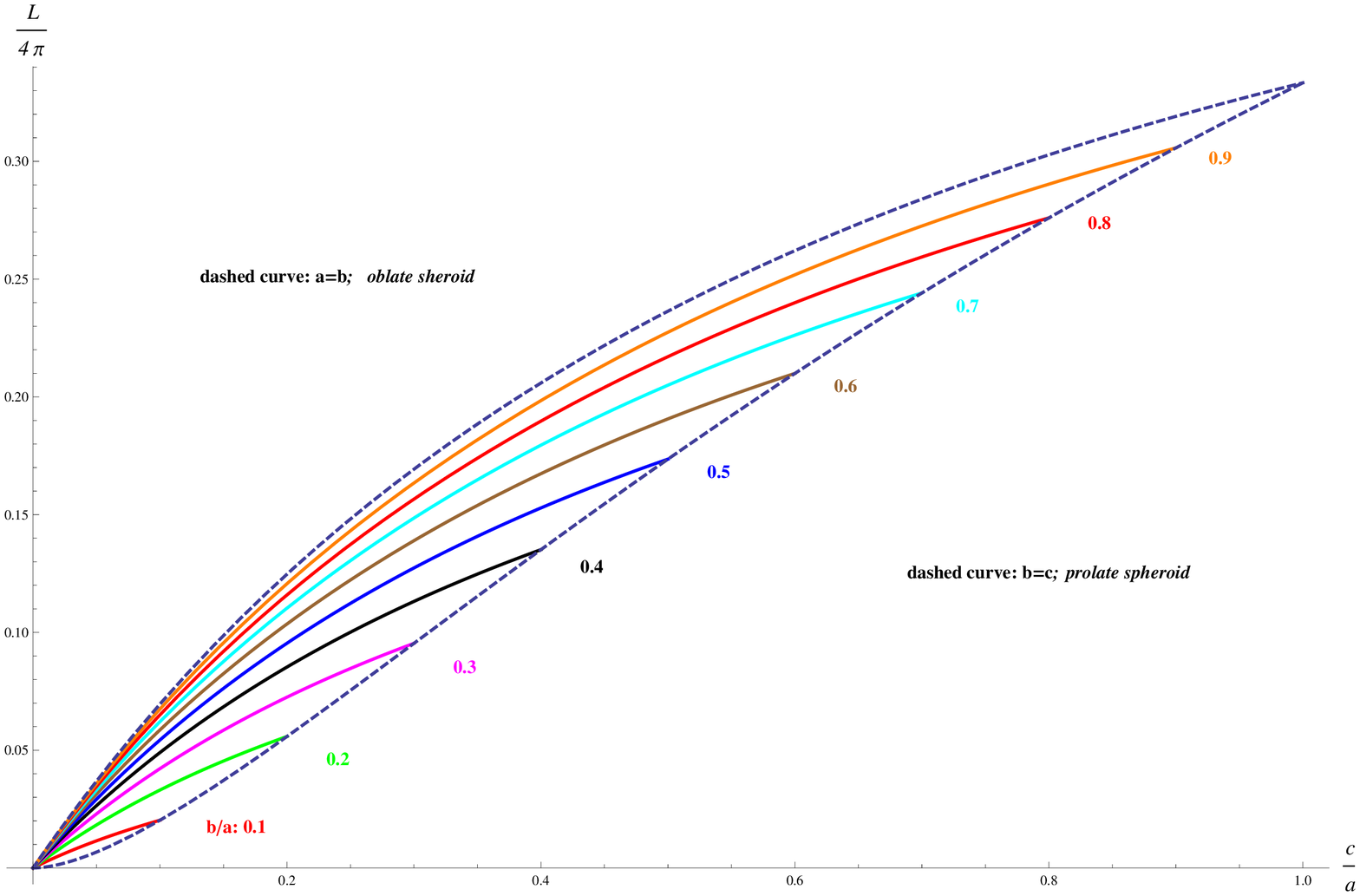}%
\caption{The $L-$ demagnetizing factor versus the ratio $c/a$ for various
values of the ratio $b/a.$}%
\label{LcdemFac}%
\end{center}
\end{figure}
%

\begin{figure}
[ptb]
\begin{center}
\includegraphics[
height=3.9254in,
width=5.7588in
]%
{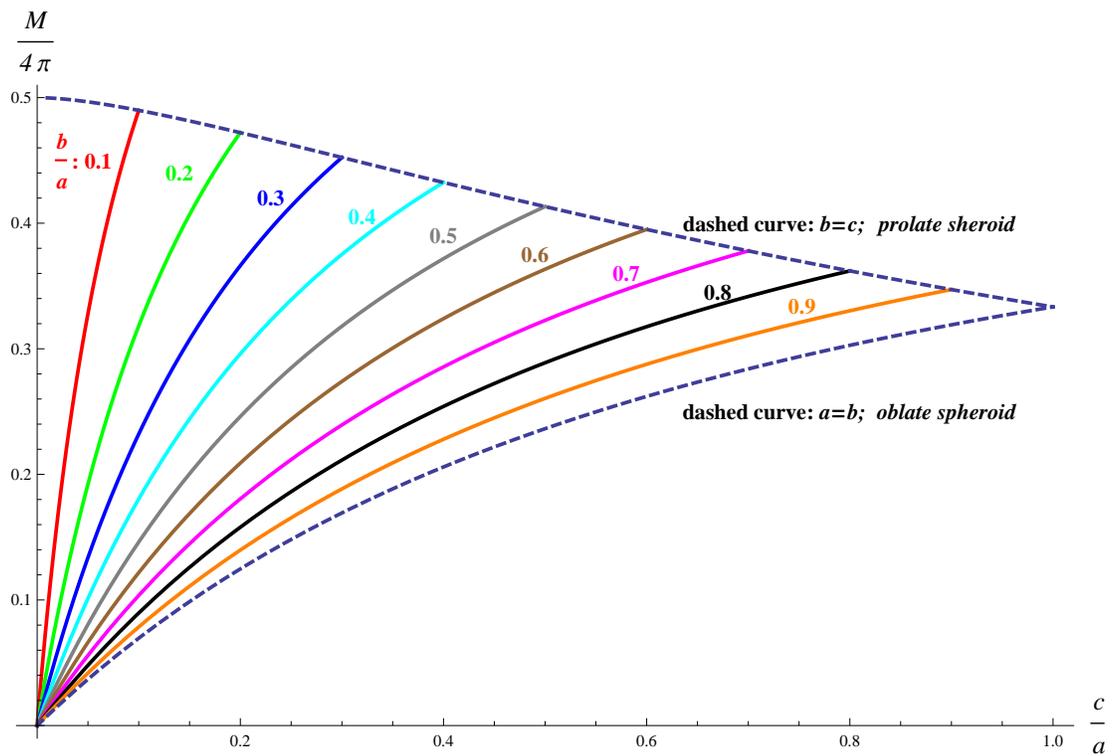}%
\caption{The $M-$demagnetizing factor versus the ratio $c/a$ for various
values of the ratio $b/a.$ The dashed curves meet at the point determined in
Corollary \ref{magsphere}.}%
\label{MdemFac1}%
\end{center}
\end{figure}

\bigskip

\bigskip%

\begin{figure}
[ptb]
\begin{center}
\includegraphics[
height=3.9435in,
width=6.0027in
]%
{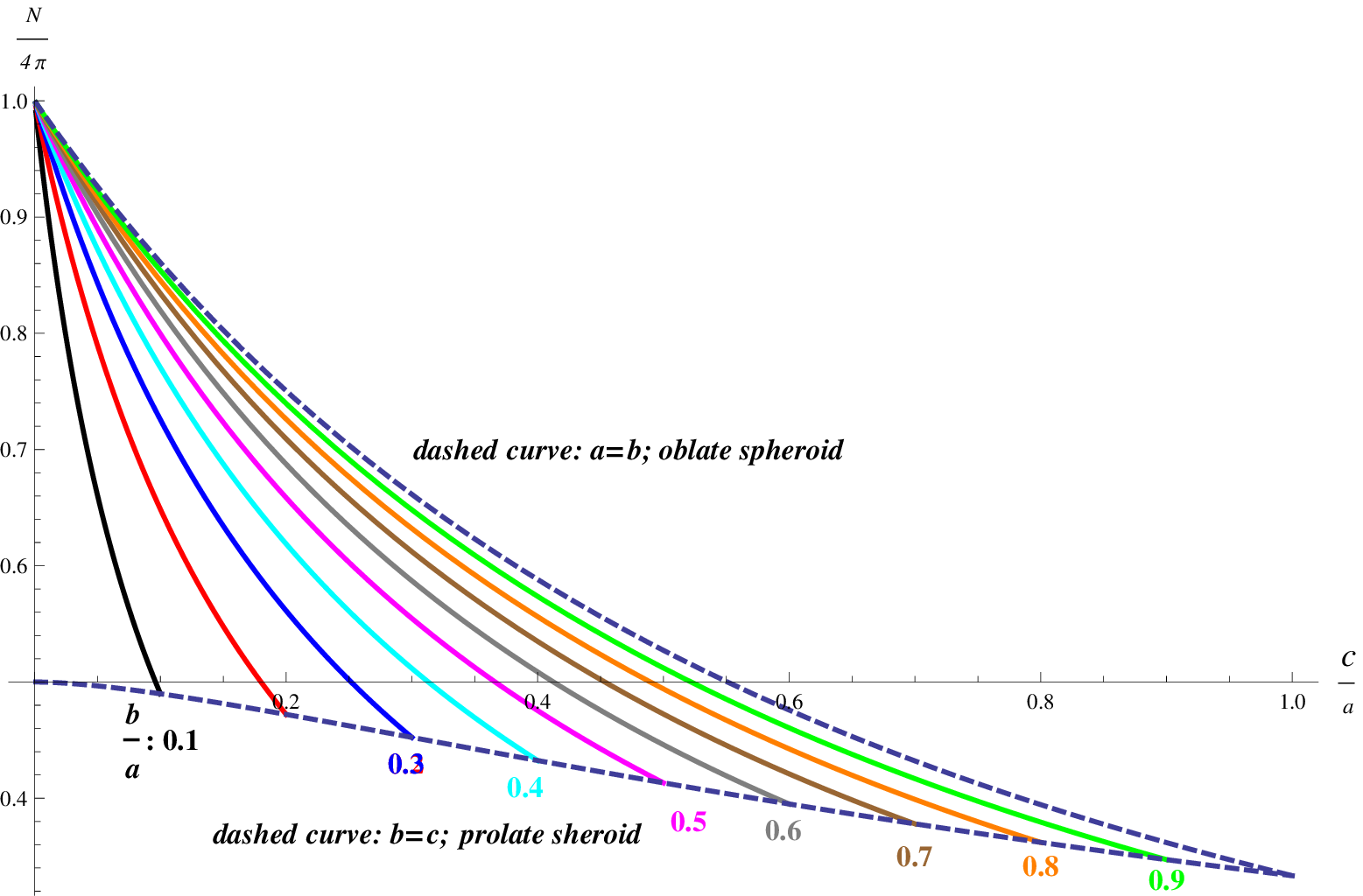}%
\caption{The $N$ demagnetizing factor versus the ratio $c/a$ for various
values of the ratio $b/a.$}%
\label{NdemFac}%
\end{center}
\end{figure}

\bigskip Using theorem \ref{magnetization} we can write for the potential
$\Phi_{int}:$%
\begin{equation}
\Phi_{int}=\mathcal{A}\left\{  1-\left(  \frac{x^{2}}{\alpha^{2}}+\frac{y^{2}%
}{\beta^{2}}+\frac{z^{2}}{\gamma^{2}}\right)  \right\}  ,
\end{equation}
where%
\begin{align}
\alpha & :=3\frac{F_{1}\left(  \frac{1}{2},\frac{1}{2},\frac{1}{2},\frac{3}%
{2},1-\frac{b^{2}}{a^{2}},1-\frac{c^{2}}{a^{2}}\right)  }{F_{1}\left(
\frac{3}{2},\frac{1}{2},\frac{1}{2},\frac{5}{2},1-\frac{b^{2}}{a^{2}}%
,1-\frac{c^{2}}{a^{2}}\right)  }a^{2},\\
\beta & :=3\frac{F_{1}\left(  \frac{1}{2},\frac{1}{2},\frac{1}{2},\frac{3}%
{2},1-\frac{b^{2}}{a^{2}},1-\frac{c^{2}}{a^{2}}\right)  }{F_{1}\left(
\frac{3}{2},\frac{3}{2},\frac{1}{2},\frac{5}{2},1-\frac{b^{2}}{a^{2}}%
,1-\frac{c^{2}}{a^{2}}\right)  }b^{2},\\
\gamma & :=3\frac{F_{1}\left(  \frac{1}{2},\frac{1}{2},\frac{1}{2},\frac{3}%
{2},1-\frac{b^{2}}{a^{2}},1-\frac{c^{2}}{a^{2}}\right)  }{F_{1}\left(
\frac{3}{2},\frac{1}{2},\frac{3}{2},\frac{5}{2},1-\frac{b^{2}}{a^{2}}%
,1-\frac{c^{2}}{a^{2}}\right)  }c^{2},\\
\mathcal{A}  & :=2\kappa\pi\frac{b}{a}\frac{c}{a}a^{2}F_{1}\left(  \frac{1}%
{2},\frac{1}{2},\frac{1}{2},\frac{3}{2},1-\frac{b^{2}}{a^{2}},1-\frac{c^{2}%
}{a^{2}}\right)  .
\end{align}

\section{\bigskip Concluding remarks}

We have solved analytically a number of problems related to the geometrical
and physical properties of the theory of ellipsoid.

In particular, we have solved in closed analytic form for the capacitance of a
conducting ellipsoid immersed in the Euclidean space $%
\mathbb{R}
^{3}.$ The exact solution has been expressed in terms of Appell's first
hypergeometric function $F_{1}$ and it is given by Theorem \ref{XWRITIKOTITA}
and equation \ $(%
\ref{CapaciGVKellipsoid}%
).$ We have also computed exactly the capacitance of a conducting ellipsoid in
$n-$dimensions. The resulting exact analytic solution is expressed in terms of
Lauricella's fourth hypergeometric function $F_{D}$ of $n-1$-variables, see
Theorem \ref{CondenserNGVK}, equation $(\ref{GVKcapacityND}). $

We subsequently solved analytically for the demagnetizing factors of a
magnetized ellipsoid immersed in the Euclidean space $%
\mathbb{R}
^{3}.$ The resulting solutions are expressed elegantly in terms of Appell's
first hypergeometric function $F_{1},$ as stated in Theorem
\ref{magnetization} and equations $(\ref{magnitisena})-(\ref{magnitis3}).$

Finally, we have derived the closed form solution for the geometrical entity
of the surface area of the ellipsoid immersed in the Euclidean space $%
\mathbb{R}
^{3}.$ Our analytic solution in this case is given in Theorem
\ref{GeorgiosVKraniotis}, eqn.$(\ref{Area51eLL13KRANIOTIS}).$

We believe that the useful exact analytic theory of the ellipsoid we have
developed in this work will have many applications in various scientific
fields. We have already outlined in the introduction possible
multidisciplinary applications of our theory in science; a scientific
multidisciplinarity which measures from physics, biology \ and chemistry to
micromechanics, space science and astrobiology. \

A fundamental mathematical generalization of our theory would be to
investigate the immersion of an ellipsoid in curved spaces \footnote{Some
initial steps along this direction have been taken in \cite{Krasinski}%
,\cite{Zsigrai}.}and solve for the corresponding geometrical and physical
properties of such an object. However, such a project is beyond the scope of
the present paper and it will be a subject of a future publication.

\bigskip

\appendix

\section{Appendix}

\bigskip The generalization of Theorem \ref{XWRITIKOTITA} for the capacitance
of the ellipsoid in $n$-dimensions involves the analytic computation of a
hyperelliptic integral. Hyperelliptic integrals which are involved in the
solution of timelike and null geodesics in Kerr and Kerr-(anti) de Sitter
black hole spacetimes have been computed analytically in references
\cite{KRANIOTISGTR},\cite{KRANIOTISLensGrav}, in terms of the multivariable
Lauricella's hypergeometric function $F_{D}$. The idea is to bring a
hyperelliptic integral by the appropriate transformations onto the integral
representation that the function $F_{D}$ admits \footnote{For an application
of the method in the realm of number theory see \cite{Giovanni}.}.

Applying this method for the analytic computation of the capacitance of the
ellipsoid in $n$-dimensions, generalizes Theorem \ref{XWRITIKOTITA} to the
following one:

\begin{theorem}
\label{CondenserNGVK}The closed form solution for the capacitance of a
$n-$dimensional conducting ellipsoid is given by the formula:%
\begin{equation}
C=\frac{2a_{1}^{n-2}\Gamma\left(  \frac{n}{2}\right)  }{\Gamma\left(
\frac{n-2}{2}\right)  \Gamma(1)}\frac{1}{F_{D}%
\left(%
\frac{n-2}{2},\underset{n-1}{\underbrace{\frac{1}{2},\frac{1}{2},\ldots
,\frac{1}{2}}},\frac{n}{2},x_{2},x_{3},\ldots,x_{n}%
\right)%
}\label{GVKcapacityND}%
\end{equation}
where $F_{D}$ denotes the fourth hypergeometric function of Lauricella of
$n-1$-variables and
\begin{equation}
x_{2}:=1-\frac{a_{2}^{2}}{a_{1}^{2}},x_{3}:=1-\frac{a_{3}^{2}}{a_{1}^{2}%
},\ldots,x_{n}:=1-\frac{a_{n}^{2}}{a_{1}^{2}}.
\end{equation}

\end{theorem}

\begin{proof}
Applying the transformation (\ref{Metasximatismos}) to the hyperelliptic
integral
\begin{equation}
\frac{1}{C}=\frac{1}{2}\int_{0}^{\infty}\frac{\mathrm{d}u}{\sqrt{(a_{1}%
^{2}+u)(a_{2}^{2}+u)\cdots(a_{n}^{2}+u)}},
\end{equation}
yields:%
\begin{align}
\frac{1}{C}  & =\frac{1}{a_{1}^{n-2}}\int_{0}^{1}\frac{x^{n-3}\mathrm{d}%
x}{\sqrt[2]{(1-\mu_{2}x^{2})(1-\mu_{3}x^{2})\cdots(1-\mu_{n}x^{2})}%
}\nonumber\\
& \overset{x^{2}=\xi}{=}\frac{1}{2a_{1}^{n-2}}\int_{0}^{1}\frac{\xi
^{\frac{n-3-1}{2}}\mathrm{d}\xi}{\sqrt[2]{(1-\mu_{2}\xi)(1-\mu_{3}\xi
)\cdots(1-\mu_{n}\xi)}}\nonumber\\
& =\frac{1}{2a_{1}^{n-2}}\frac{\Gamma\left(  \frac{n-2}{2}\right)  \Gamma
(1)}{\Gamma\left(  \frac{n}{2}\right)  }F_{D}%
\left(%
\frac{n-2}{2},\underset{n-1}{\underbrace{\frac{1}{2},\frac{1}{2},\ldots
,\frac{1}{2}}},\frac{n}{2},x_{2},x_{3},\ldots,x_{n}%
\right)%
\end{align}
where
\begin{equation}
\mu_{j}\equiv x_{j}:=1-\frac{a_{j}^{2}}{a_{1}^{2}},j=1,2,\ldots,n,
\end{equation}
and
\begin{equation}
a_{1}\geq a_{2}\geq\ldots\geq a_{n}>0
\end{equation}

For $n=3$ we derive
\begin{equation}
\frac{1}{C}=\frac{1}{2a_{1}}\frac{\Gamma(1/2)\Gamma(1)}{\Gamma(3/2)}%
F_{1}\left(  \frac{1}{2},\frac{1}{2},\frac{1}{2},\frac{3}{2},1-\frac{b^{2}%
}{a^{2}},1-\frac{c^{2}}{a^{2}}\right)  ,
\end{equation}
and thus Theorem \ref{XWRITIKOTITA} is proved as well.
\end{proof}

Applying our closed form analytic formula, eqn.(\ref{GVKcapacityND}), for
$n=4,$ and the choice of values $a_{1}=2,a_{2}=1+2/3,a_{2}=1+1/3,a_{4}=1$ we
derive for the capacitance of this particular higher dimensional ellipsoid:%
\begin{equation}
C_{(n=4)}=4.406592791665676649174487.
\end{equation}

The generalization of Theorem \ref{magnetization} in $n$-dimensions is the following:

\begin{theorem}%
\begin{align}
\mathcal{L}^{(n)}  & =\frac{\pi a_{1}a_{2}a_{3}\ldots a_{n}}{a_{1}^{n}}%
\frac{\Gamma\left(  \frac{n}{2}\right)  \Gamma(1)}{\Gamma\left(  \frac{n+2}%
{2}\right)  }F_{D}\left(  \frac{n}{2},\underset{n-1}{\underbrace{\frac{1}%
{2},\frac{1}{2},\ldots,\frac{1}{2}}},\frac{n+2}{2},\mu_{2},\ldots,\mu
_{n}\right)  ,\nonumber\\
& \\
\mathcal{M}^{(n)}  & =\frac{\pi a_{1}a_{2}a_{3}\ldots a_{n}}{a_{1}^{n}}%
\frac{\Gamma\left(  \frac{n}{2}\right)  \Gamma(1)}{\Gamma\left(  \frac{n+2}%
{2}\right)  }F_{D}\left(  \frac{n}{2},\frac{3}{2},\frac{1}{2},\ldots,\frac
{1}{2},\frac{n+2}{2},\mu_{2},\ldots,\mu_{n}\right)  ,\nonumber\\
&
\end{align}

\end{theorem}

\bigskip

\bigskip The fourth hypergeometric function of Lauricella $F_{D}$ of
$m$-variables \cite{GLAURICELLA} is defined as follows:%

\begin{equation}
\fbox{$\displaystyle F_D(\alpha,\mbox{\boldmath${\beta}$},\gamma,{\bf
z})\equiv F_D^{(m)}(\alpha,\mbox{\boldmath${\beta}$},\gamma,{\bf z}):=
\sum_{n_1,n_2,\dots,n_m=0}^{\infty}\frac{(\alpha)_{n_1+\cdots n_m}%
(\beta_1)_{n_1} \cdots(\beta_m)_{n_m}}
{(\gamma)_{n_1+\cdots+n_m}(1)_{n_1}\cdots(1)_{n_m}} z_1^{n_1}\cdots
z_m^{n_m}$}
\label{GLauri}
\end{equation}%

where%

\begin{eqnarray}
\mathbf{z} &=&(z_{1},\ldots,z_{m}),  \notag\\
\mbox{\boldmath${\beta}$} &=&(\beta_{1},\ldots,\beta_{m}).
\end{eqnarray}
The Pochhammer symbol
\fbox{$\displaystyle(\alpha)_m=(\alpha,m)$}
is defined by%

\begin{equation}
(\alpha)_{m}=\frac{\Gamma(\alpha+m)}{\Gamma(\alpha)}=\left\{
\begin{array}{ccc}
1, &
{\rm if}
& m=0 \\
\alpha(\alpha+1)\cdots(\alpha+m-1) & \text{{\rm if}} & m=1,2,3\end{array}%
\right.
\end{equation}

The series admits the following integral representation:%
\begin{equation}
\fbox{$\displaystyle F_D(\alpha,\mbox{\boldmath${\beta}$},\gamma,{\bf z})=
\frac{\Gamma(\gamma)}{\Gamma(\alpha)\Gamma(\gamma-\alpha)}
\int_0^1 t^{\alpha-1}(1-t)^{\gamma-\alpha-1}(1-z_1 t)^{-\beta_1}%
\cdots(1-z_m t)^{-\beta_m} {\rm d}t $}
\label{OloklAnapa}
\end{equation}
which is valid for
\fbox{$\displaystyle{\rm Re}(\alpha)>0,\;{\rm Re}(\gamma-\alpha)>0. $}%
. It converges absolutely inside the $m$-dimensional cuboid%

\begin{equation}
|z_{j}|<1,(j=1,\ldots,m).
\end{equation}

It also has the following values:%
\begin{align}
& F_{D}^{(n)}(\alpha,\beta_{1},\ldots,\beta_{n},\gamma,1,x_{2},\ldots
,x_{n})\nonumber\\
& =\frac{\Gamma(\gamma)\Gamma(\gamma-\alpha-\beta_{1})}{\Gamma(\gamma
-\alpha)\Gamma(\gamma-\beta_{1})}F_{D}^{(n-1)}(\alpha,\beta_{2},\ldots
,\beta_{n},\gamma-\beta_{1},x_{2},\ldots,x_{n}),
\end{align}
when $\max\{|x_{2}|,\ldots,|x_{n}|\}<1,$ $\operatorname{Re}(\gamma
-\alpha-\beta_{1})>0.$

\end{document}